\documentclass[conference]{IEEEtran}
\usepackage{ifpdf}

\usepackage{cite}

\ifCLASSINFOpdf
  \usepackage[pdftex]{graphicx}
\else
\fi
\usepackage{amsmath}
\usepackage{amsfonts}
\usepackage{amssymb}
\usepackage{amsthm}
\usepackage{color}

\DeclareMathOperator*{\argmin}{arg\,min}
\usepackage{algorithmic}
\usepackage{array}
\usepackage{url}
\newtheorem{theorem}{Theorem}
\newtheorem{definition}{Definition}
\newtheorem{example}{Example}

\newtheorem{corollary}{Corollary}
\newtheorem{lemma}{Lemma}
\newtheorem{remark}{Remark}
\newtheorem{protocol}{Protocol}

\newcommand{\cV}{\mathcal{V}}
\newcommand{\cB}{\mathcal{B}}

\begin{document}
%
\title{On the Communication Cost of Determining an Approximate Nearest Lattice Point}




%
\author{\IEEEauthorblockN{Maiara F. Bollauf \IEEEauthorrefmark{1},
Vinay A. Vaishampayan \IEEEauthorrefmark{2} and
Sueli I. R. Costa \IEEEauthorrefmark{3}}
\vspace{0.2cm}
\IEEEauthorblockA{\IEEEauthorrefmark{1} \IEEEauthorrefmark{3} Institute of Mathematics, Statistic and Computer Science\\
University of Campinas, Sao Paulo, Brazil\\ Email: maiarabollauf@ime.unicamp.br, sueli@ime.unicamp.br}
\vspace{0.2cm}
\IEEEauthorblockA{\IEEEauthorrefmark{2}Department of Engineering Science and Physics \\
College of Staten Island, City University of New York, Staten Island - NY, United States \\ Email: Vinay.Vaishampayan@csi.cuny.edu}}



\maketitle

\begin{abstract}
We consider the closest lattice point problem in a distributed network setting and study  the communication cost and the error probability for computing an approximate nearest lattice point, using the nearest-plane algorithm, due to Babai. Two distinct communication models, centralized and interactive, are considered. The importance of proper basis selection is addressed. Assuming a reduced basis for a  two-dimensional lattice, we determine the approximation error of the nearest plane algorithm.  The communication cost  for determining the Babai point, or equivalently, for constructing  the rectangular nearest-plane partition, is calculated in the interactive setting. For the centralized model, an algorithm is presented for reducing the communication cost  of the nearest plane algorithm in an arbitrary number of dimensions.    

\end{abstract}



%
\IEEEpeerreviewmaketitle

\section{Introduction}


	
	
	\begin{figure}[h]
\centering
\begin{minipage}{.45\linewidth}
  \includegraphics[width=\linewidth]{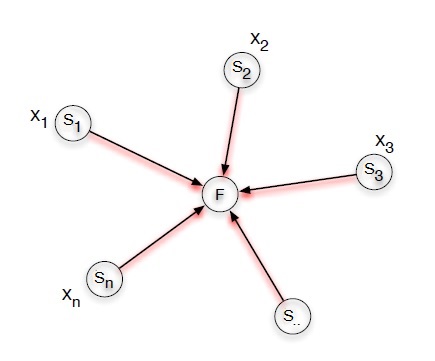}
  \caption{Centralized model}
  \label{fig-p1}
\end{minipage}
\hspace{.05\linewidth}
\begin{minipage}{.45\linewidth}
  \includegraphics[width=\linewidth]{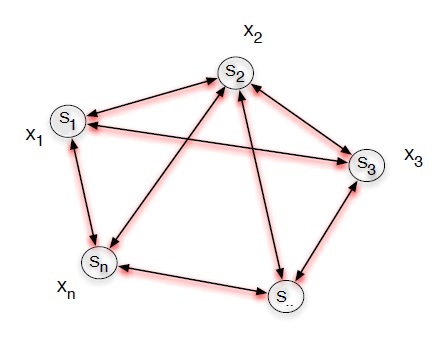}
  \caption{Interactive model}
  \label{fig-p2}
\end{minipage}
\end{figure}

A network consists of $N$ sensor-processor nodes (hereafter referred to as nodes) and possibly a central computing node (fusion center) $F$ interconnected by links  with limited bandwidth. 
Node $i$ observes real-valued random variable $X_i$. In the  \textit{centralized model} (Fig. \ref{fig-p1}), the objective is to compute a given function $f(X_1,X_2,\ldots,X_n)$ at the fusion center based on information communicated from each of the $N$ sensor nodes. In the \textit{interactive model} (Fig. \ref{fig-p2}), the objective is to compute the function $f(X_1,X_2,\ldots,X_n)$ at each sensor node (the fusion center is absent). In general, since the random variables are real valued, these calculations would require that the system communicate an infinite number of bits in order to compute $f$ exactly. Since the network has finite bandwidth links, the information must be quantized in a suitable manner, but quantization affects the accuracy of the function that we are trying to compute. Thus, the main goal is to manage the tradeoff between communication cost and function computation accuracy.
	
In this work,  $f$ computes the closest lattice point  to a real vector $x=(x_1,x_2,\ldots,x_n)$ in a given lattice $\Lambda$. The process of finding the closest lattice point is widely used for decoding lattice codes, and for quantization.  Lattice coding  offers significant coding gains for noisy channel communication~\cite{conwaysloane}  and for quantization~\cite{berger}. In a network, it may be necessary for a vector of measurements to be available at locations other than and possibly including nodes where the measurements are made. In order to reduce network bandwidth usage, it is logical to consider a vector quantized (VQ) representation of these measurements, subject to a fidelity criterion, for once a VQ representation is obtained, it can be forwarded in a bandwidth efficient manner to other parts of the network. However, there is a communication cost to obtaining the vector quantized representation. This paper is our attempt to understand the costs and tradeoffs involved. Example application settings include MIMO systems~\cite{RCP:2009}, and network management in wide area networks~\cite{KCR:2006}, to name a few. For prior work in the computer science community, see~\cite{Yao:1979},~\cite{KN:1997}. Information theory~[\ref{Shannon:1948}] has resulted in tight bounds,~\cite{Orlitsky:2001},~\cite{MaIshwar}.  The communication cost/error tradeoff  of refining the nearest-plane estimate obtained here is addressed in a companion paper~\cite{VB:2017}.

We observe here that algorithms for the closest lattice point problem have been studied in great detail, see~\cite{agrelletal} and the references therein, for a comprehensive survey and novel algorithms. However, in all these algorithms it is assumed that the vector components are available at the same location. In our work, the vector components are available at physically separated nodes and we are interested in the communication cost of exchanging this information in order to determine the closest lattice point. None of the previously proposed fast algorithms consider this communication cost.

The remainder of our paper is organized into three sections:  Sec.~\ref{sec1} presents some basic definitions,  and establishes a framework for measuring the cost and error rate. Sec.~\ref{sec3} presents an expression for the probability of error of the distributed closest lattice point problem in an arbitrary two-dimensional case, Sec.~\ref{sec:babairate} presents rate estimates for both models for arbitrary $n>1$ and Sec.~\ref{secCC} presents conclusions and directions for future work.

\section{Lattice Basics, Voronoi and Babai Partitions} 
\label{sec1} 
Notation and essential aspects of lattice coding are described in this section.

\begin{definition}(\textit{Lattice}) A lattice $\Lambda \subset \mathbb{R}^M$ is the set of integer linear combinations of independent vectors $v_1,v_2,\ldots,v_n \in \mathbb{R}^M$ , with $n \leq M$,
\end{definition}  

\begin{definition}(\textit{Generator matrix}) The generator matrix of the lattice is represented by matrix $V$ with $i$th column  $v_i$, $i=1,2,\ldots,n$. Thus $\Lambda=\{Vu,~u\in \mathbb{Z}^n\},$ where $u$ is considered here as a column vector.
\end{definition}

We will assume in the sequence of our work that $\Lambda$ has full rank ($n=M$).

\begin{definition}(\textit{Voronoi cell}) The Voronoi cell $\cV(\lambda)$ of a lattice $\Lambda \subset \mathbb{R}^{n}$ is the subset of $\mathbb{R}^{n}$ containing all points nearer to lattice point $\lambda$ than to any other lattice point:
\begin{equation}
\cV(\lambda)=\{x \in \mathbb{R}^{n}: ||x-\lambda|| \leq ||x-\tilde{\lambda}||, \ \text{for all} \ \tilde{\lambda} \in \Lambda\},
\end{equation}
where $||.||$ denotes the Euclidean norm.
\end{definition}

\begin{definition}(\textit{Relevant vector}) A vector $v$ is said to be a relevant vector of a lattice $\Lambda$ if the intersection of the hyperplane $\{x \in \mathbb{R}^{n}: \langle x,v \rangle=\frac{1}{2}\langle v,v \rangle\}$ with $\cV(0)$ is an $(n-1)-$dimensional face of $\cV(0).$
\end{definition}

The closest vector problem (CVP) in a lattice can be described as an integer least squares problem with the objective of determining $u^*,$ such that
\begin{equation} \label{ilsp}
u^\ast =  \argmin_{u \in \mathbb{Z}^{n}} \mid\mid x-Vu \mid\mid^{2},
\end{equation}
where the norm considered is the standard Euclidean norm. The closest lattice point to $x$ is then given by $x_{nl}=Vu^\ast$. The mapping $g_{nl}~:~\mathbb{R}^n \rightarrow \Lambda,$ $x \mapsto x_{nl}$ partitions $\mathbb{R}^n$ into  Voronoi cells, each of volume $|\det V|$.

The nearest plane (np) algorithm computes $x_{np}$, an approximation to $x_{nl}$, given by $x_{np}=b_1 v_1 + b_2 v_2+\ldots+b_n v_n$, where $b_i \in \mathbb{Z}$ is obtained as follows, derived from \cite{babai}. 

Let ${\mathcal S}_i$ denote the subspace spanned by the vectors $\{v_1,v_2,\ldots,v_i\}$, $i=1,2,\ldots,n$. Let ${\mathcal P}_i(z)$ be the orthogonal projection of $z$ onto ${\mathcal S}_i$ and let $v_{i,i-1}={\mathcal P}_{i-1}(v_i)$ be the nearest vector to $v_i$ in ${\mathcal S}_{i-1}$.
We have the following unique decomposition: $v_i=v_{i,i-1}+v_{i,i-1}^{\perp}$. Also,  let $z_i^\perp=z_{i}-{\mathcal P}_i(z_{i})$. Start with $z_n=x$ and $i=n$ and compute $b_i=\left[ \langle z_{i},v_{i,i-1}^\perp \rangle/\|v_{i,i-1}^\perp\|^2 \right]$, $z_{i-1}={\mathcal P}_{i-1}(z_i)-b_i v_{i,i-1}$, for $i=n,n-1,\ldots,1$. Here $[x]$ denotes the nearest integer to $x$

The mapping $g_{np}~:~\mathbb{R}^n \rightarrow \Lambda,$ $x \mapsto x_{np},$ partitions $\mathbb{R}^n$ into  hyper-rectangular cells with volume $|\det V|$, as illustrated in Fig.~\ref{fig-np} for the hexagonal lattice $A_2$. We refer to this partition as a \emph{Babai} partition. Note that this partition is basis dependent. In case $V$ is upper triangular with $(i,j)$ entry $v_{ij}$, each rectangular cell is axis-aligned and has sides of length  $|v_{11}|,|v_{22}|,\ldots,|v_{nn}|$. 

\begin{figure}[h!]
\begin{center}
		\includegraphics[height=3cm]{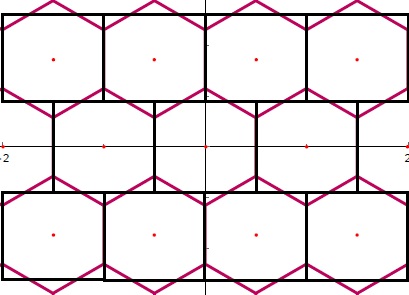}  
\caption{{Cells of the  \emph{np} or \textit{Babai} partition (black boundaries) and the Voronoi partition (pink solid lines) of $\mathbb{R}^{2}$ for hexagonal lattice $A_2$ with basis $\{(1,0),(1/2,\sqrt{3}/2)\}$}}
 \label{fig-np}
\end{center}
\end{figure}
\vspace{-0.15in}

\begin{definition} (Minkowski-reduced basis~\cite{mink}) A basis $\{v_{1},v_{2},...,v_{n}\}$ of a lattice $\Lambda$
in $\mathbb{R}^{n}$ is said to be Minkowski-reduced
if $v_{j},$ with $j=1,\dots,n,$ is such that $\left\Vert v_{j}\right\Vert \leq\left\Vert v\right\Vert $, for any $v$ for which $\{v_{1},...,v_{j-1},v\}$ can be extended
to a basis of $\Lambda$.
\end{definition}

In particular, for lattices of dimension $n \leq 4,$ the norms of the Minkowski-reduced basis vectors achieve the successive minima \cite{pohst}.
For two-dimensional lattices, a Minkowski-reduced
basis is also called Lagrange-Gauss reduced basis and there
is a simple characterization~\cite{conwaysloane}: 
a lattice basis $\left\{ \ensuremath{v_{1},v_{2}}\right\} $
is a Minkowski-reduced basis if only if $\left\Vert v_{1}\right\Vert \leq\left\Vert v_{2}\right\Vert $
and $2\langle v_{1},v_{2}\rangle  \leq \left\Vert v_{1}\right\Vert ^{2}.$ 
It follows that the angle $\theta$ between
the minimum norm vectors $v_{1}$ and $v_{2}$ must satisfy $\text{ }\frac{\pi}{3} \leq\theta\leq\frac{2\pi}{3}.$

Since a Minkowski-reduced basis consists of short vectors that are  ``as  perpendicular as possible", it is a good choice for starting the np-algorithm. But
it is computationally hard to get such a basis from an arbitrary one.
One alternative is to use the basis obtained with the LLL algorithm~\cite{lll}, which approximates the Minkowski
basis and can be achieved in polynomial time.
For a basis that is LLL reduced, the ratio of the distances $\|x-x_{np}\|/\|x-x_{nl}\|$ can be bounded above by a constant that depends on the dimension alone~\cite{babai}. 

\section{Error Probability Analysis for an arbitrary Two-Dimensional Lattice}
\label{sec3}
We  assume that node-$i$ observes an independent identically distributed (iid) random process $\{X_i(t), t \in \mathbb{Z}\}$, where $t$ is the time index and that random processes observed at distinct nodes are mutually independent. The time index $t$ is suppressed in the sequel. The random vector $X=(X_1,X_2)$ is obtained by projecting a random process on the basis vectors of an underlying coordinate frame, which is assumed to be fixed.

Consider that the lattice $\Lambda$ is generated by the scaled generator matrix $\alpha V$, where $V$ is the generator matrix of the unscaled lattice. 
Let ${\mathcal V}(\lambda)$ and ${\mathcal B}(\lambda)$ denote the Voronoi and Babai cells, respectively, associated with lattice vector $\lambda \in \Lambda$. The error probability $P_e(\alpha)$, is the probability of the event $\{\lambda_{nl}(X)\neq \lambda_{np}(X)\}$ and $P_e:=\lim_{\alpha \rightarrow 0} P_e(\alpha) = {area(\cB(0)\bigcap \cV(0)^c)}/{area(\cB(0))}$. 

As will be discussed in this section, the Babai partition is dependent on, and the Voronoi partition is invariant to, the choice of lattice basis. Thus the error probability depends on the choice of the lattice basis. We will assume here that a Minkowski-reduced lattice basis can be chosen by the designer of the lattice code and it can be transformed into an equivalent basis $\{(1,0),(a,b)\}.$ This can be accomplished by applying QR decomposition to the lattice generator matrix (which has the original chosen basis vectors on its columns) in addition to convenient scalar factor. The reason for working with a Minkowski-reduced basis is partly justified by Ex.~\ref{ex1} below and the fact that the Voronoi region is easily determined since the relevant vectors are known; see Lemma~\ref{lemmathird} below.




An example to demonstrate the dependence of the error probability on the lattice basis is now presented.
\begin{example} \label{ex1} Consider a lattice
$\Lambda\subset\mathbb{R}^{2}$ with basis $\{(5,0),(3,1)\}.$ The probability of error in this case is $P_{e}=0.6$ (Fig. \ref{triangbasis}), whereas if we
start from the basis $\{(1,2),(-2,1)\},$ we achieve after the QR decomposition $\left\{ (\sqrt{5},0),(0,\sqrt{5})\right\}$ and $P_{e}=0,$ since the Babai region associated
with an orthogonal basis and the Voronoi region for rectangular lattices coincides.

\begin{figure}[h!]
\begin{center}
		\includegraphics[height=3.0cm]{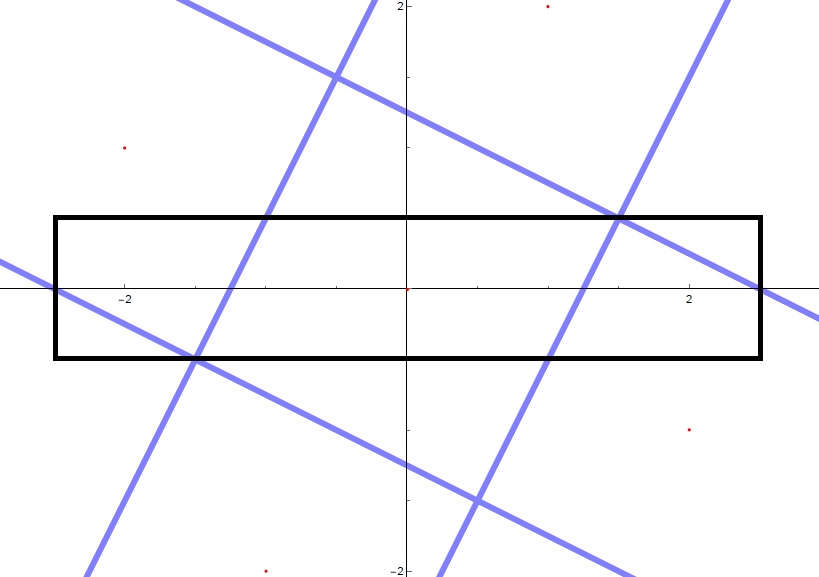}  
\caption{{Voronoi region and Babai partition of the triangular basis $\{(5,0),(3,1)\}$ }}
\label{triangbasis}
\end{center}
\end{figure}
\end{example}
Example \ref{ex1} illustrates the importance of  working with a good basis and partially explains our choice to work with a Minkowski-reduced basis. 
As mentioned above, additional motivation come from the observation that for a Minkowski-reduced basis in two dimensions, the relevant vectors are known.	

%


%

To see this, we first note that an equivalent condition for  a basis $\{v_{1}, v_{2}\}$ to be Minkowski reduced in dimension
two is $||v_{1}||\leq||v_{2}||\leq||v_{1}\pm v_{2}||$ ([\ref{galbraith}], Lemma 17.1.4),
from which the following result can be derived.

\begin{lemma} \label{lemmathird} If a Minkowski-reduced basis is given by $\{(1,0),(a,b)\}$ then, besides the basis vectors, a third relevant vector is
\begin{equation}
\begin{cases}
(-1+a,b), & \text{if } \frac{\pi}{3} \leq \theta \leq \frac{\pi}{2} \\
(1+a,b), & \text{if } \frac{\pi}{2} < \theta \leq \frac{2\pi}{3},
\end{cases}
\end{equation}
where $\theta$ is the angle between $(1,0)$ and $(a,b).$
\end{lemma}

Note that, if $\{{v\ensuremath{_{1},v_{2}}}\}$ is a Minkowski basis
then so is $\{{-v\ensuremath{_{1},v_{2}}}\}$ and hence any lattice
has a Minkowski basis with $\frac{\pi}{3}\leq\theta\leq\frac{\pi}{2}$.
So, if we consider the Minkowski-reduced basis as $\{(1,0),(a,b)\},$ with $a^{2}+b^{2} \geq 1$ and $0 \leq a \leq \frac{1}{2},$ it is possible to use Lemma \ref{lemmathird} to describe the Voronoi region of $\Lambda$ and determine its intersection with the associated Babai partition. Observe that the area of both regions must be the same and in this specific case, equal to $b.$ This means that the vertices that define the Babai rectangular partition are $\left(\pm \frac{1}{2}, \pm \frac{b}{2}\right).$ Therefore, we can state the following result

\begin{theorem} \label{thmmain} Consider a lattice $\Lambda \subset\mathbb{R}^{2}$ with a triangular Minkowski-reduced basis $\beta=\{v_1,v_2\}=\{(1,0),(a,b)\}$   such that the angle $\theta$ between $v_{1}$ and $v_{2}$ satisfies \textup{$\frac{\pi}{3}\leq\theta\leq\frac{\pi}{2}$.} The probability of error $P_e$ for the Babai partition is given by 
\begin{equation}
P_e=F(a, b)=\frac{a-a^{2}}{4b^{2}}.
\end{equation} 
\end{theorem}

\begin{proof} To calculate $P_{e}$ for the lattice $\Lambda$, we first obtain the vertices of the Voronoi region. This is done by calculating the points of intersection of the perpendicular bisectors of the  three relevant vectors $(1,0), (a,b)$ and $(-1+a,b)$ (according to Lemma \ref{lemmathird}, Fig. \ref{rvectors}). Thus the vertices of the Voronoi region are given by  $\pm(\frac{1}{2},\frac{a^{2}+b^{2}-a}{2b})$,
$\pm(-\frac{1}{2},\frac{a^{2}+b^{2}-a}{2b})$ and $\pm(\frac{2a-1}{2},\frac{-a^{2}+b^{2}+a}{2})$.

\begin{figure}[h!]
\begin{center}
		\includegraphics[height=4.0cm]{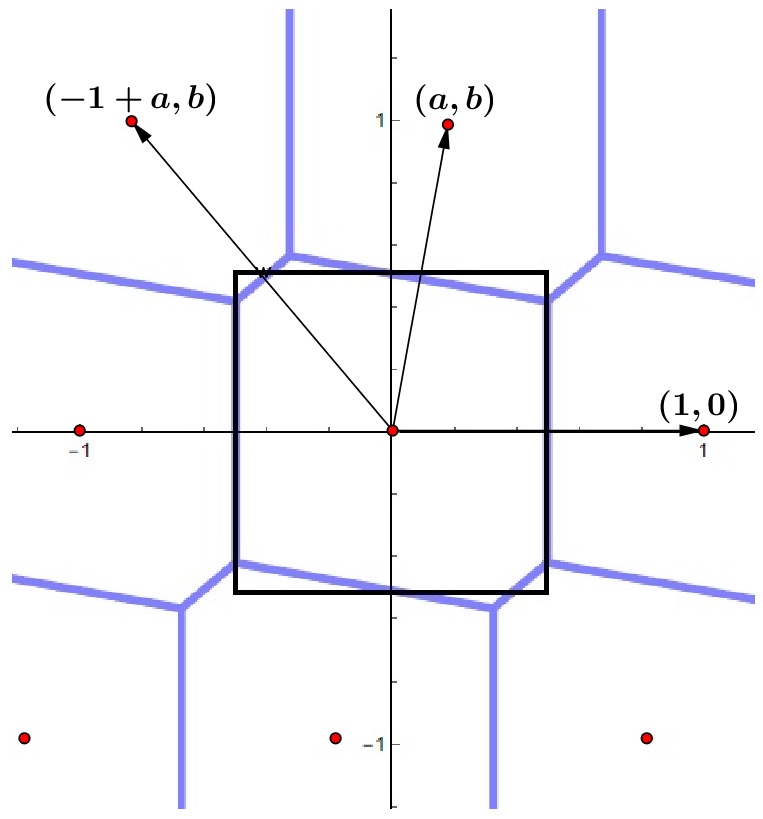}  
\caption{{Voronoi region, Babai partition and three relevant vectors}}
\label{rvectors}
\end{center}
\end{figure}

$P_{e}$ is then computed as the ratio between the area of the Babai
region which is not overlapped by the Voronoi region $\mathcal{V}(0)$ and the area $|b|$ of the Babai
region. From Fig. \ref{rvectors}, we get the error as the sum the areas of four triangles, where two of them are defined respectively by the points $\left(\frac{1}{2},\frac{b}{2}\right), \left(\frac{1}{2},\frac{a^{2}-a+b^{2}}{2b}\right), \left(\frac{a}{2}, \frac{b}{2}\right)$ and $\left(-\frac{1}{2},\frac{b}{2}\right), \left(-\frac{1}{2},\frac{a^{2}-a+b^{2}}{2b}\right), \left(\frac{a-1}{2}, \frac{b}{2}\right).$ The remaining two triangles are symmetric to these two. Therefore, the probability of error is the sum of the four areas, normalized by the area of the Voronoi region $|\det(V)|=|b|$. The explicit formula for it is given by $F(a, b) =  \dfrac{1}{4} \dfrac{a-a^{2}}{b^2}.$
\end{proof}

\begin{remark}
Note also that 
if $\rho:=\frac{\left\Vert v_{2}\right\Vert }{\left\Vert v_{1}\right\Vert }$
and  $\theta$ is defined to be the angle between the basis vectors, then the result of Theorem \ref{thmmain}
can be rewritten as 
\begin{equation}
P_{e}=H(\theta,\rho)=\frac{1}{4\rho}\frac{|\cos\theta|}{\sin^{2}\theta}(1-\rho|\cos\theta|).
\end{equation}
\end{remark}

We obtain the following Corollary, illustrated in Fig. \ref{contour},	 from  the probability of error $P_{e}=F(a,b)=\frac{1}{4} \frac{a}{b^2}(1-a)=\frac{1-(1-2a)^{2}}{16b^{2}}$ obtained in Theorem \ref{thmmain} with  $b \geq \frac{\sqrt{3}}{2}$ and $0 \leq a \leq \frac{1}{2}$.
		
\begin{corollary} For any two-dimensional lattice and a Babai partition constructed from the QR decomposition associated with a Minkowski-reduced basis where $\frac{\pi}{3} \leq \theta \leq \frac{\pi}{2},$ we have 
\begin{equation}
0 \leq P_{e} \leq \frac{1}{12},
\end{equation}
and
\begin{itemize}
\item[a)] $P_{e}=0 \Longleftrightarrow a=0,$ i.e., the lattice is orthogonal.
\item[b)] $P_{e}=\frac{1}{12} \Longleftrightarrow (a,b)=\left(\frac{1}{2}, \frac{\sqrt{3}}{2}\right),$ i.e., the lattice is equivalent to hexagonal lattice.
\item[c)] the level curves of $P_{e}$ are described as ellipsoidal arcs in the region $a^{2}+b^{2} \geq 1$ and $0 \leq a \leq \frac{1}{2}.$
\end{itemize}

\end{corollary}

\begin{figure}[h!]
\begin{center}
		\includegraphics[height=6.5cm]{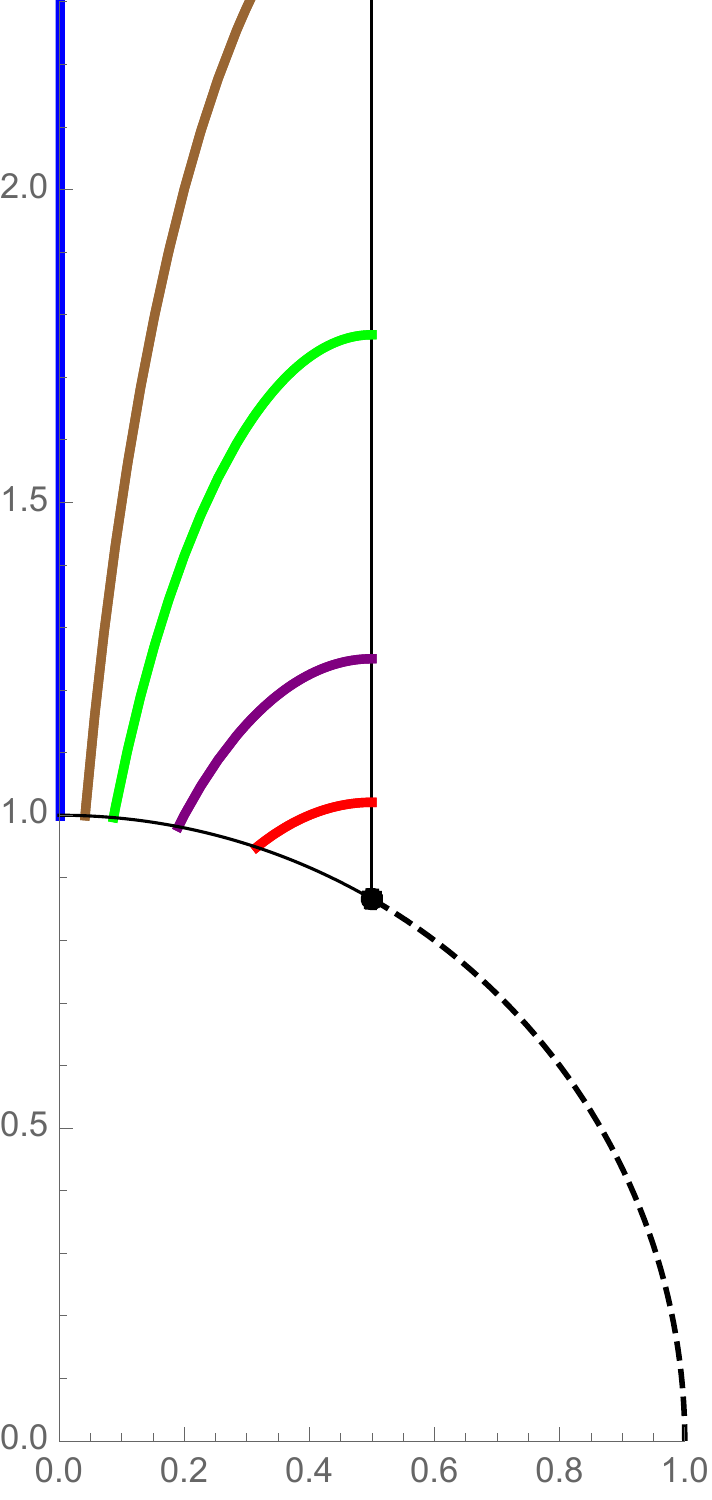}  
\caption{{Level curves of $P_{e}=k,$ in left-right ordering, for $k=0, k=0.01, k=0.02, k=0.04, k=0.06$ and $k=1/12 \approx 0.0833.$ Notice that $a$ is represented in the horizontal axis and $b$ in vertical axis. }}
\label{contour}
\end{center}
\end{figure}

\section{Rate Computation for Constructing a Babai Partition for arbitrary $n>1$}
\label{sec:babairate}
Communication protocols are presented for the centralized and interactive model along with associated rate calculations in the limit as $\alpha \rightarrow 0$.   

\subsection{Centralized Model}
We now describe the transmission protocol  $\Pi_{c}$ by which the nearest plane lattice point can be determined at the fusion center $F$.
Let $v_{m,l}/v_{m,m}=p_{m,l}/q_{m,l}$ where $p_{m,l}$ and $q_{m,l}>0$ are relatively prime. Note that we are assuming the generator matrix is such that the aforementioned ratios are rational, for $l>m.$ Let $q_m=l.c.m \ \{q_{m,l}, l>m\}$, where $l.c.m$ denotes the least common multiple of its argument. By definition $q_n=1$.
\begin{protocol} (Transmission, $\Pi_{c}$). Let $s(m)\in\{0,1,\ldots,q_m-1\}$ be the largest $s$ for which $[x_m/v_{m,m}-s/q_m]=[x_m/v_{m,m}]$. Then node $m$ sends 
$\tilde{b}_m=[x_m/v_{m,m}]$ and $s(m)$ to $F$, $m=1,2,\ldots,n$ (by definition $s(n)=0$).
\end{protocol}

Let $\overline{b}=(b_1,b_2,\ldots,b_n)$ be the coefficients of $\lambda_{np}$, the Babai point.
\begin{theorem}
The coefficients of the Babai point $\overline{b}$ can be determined at the fusion center $F$ after running transmission protocol $\Pi_c$.
\end{theorem}
\begin{proof}
Observe that each coefficient of $\overline{b}$ is given by
\begin{eqnarray}
\lefteqn{b_m=} & \nonumber \\ 
& \left[ \frac{x_m-\sum_{l=m+1}^{n} b_{l}v_{m,l}}{v_{m,m}}\right],
m=1,2,\ldots,n,
\end{eqnarray}
which is written in terms of  $\{z \}$ and $\lfloor z \rfloor$, the fractional and integer parts of real number $z$, resp.,  ($z=\lfloor z\rfloor +\{z\}$, $0 \leq \{z\} < 1$) by
\begin{eqnarray}
\lefteqn{b_m=} & \nonumber \\ 
& \left[ \frac{x_m}{v_{m,m}}-\left\{\frac{\sum_{l=m+1}^{n} b_{l}v_{m,l}}{v_{m,m}}\right\} \right]-  \left\lfloor \frac{\sum_{l=m+1}^{n} b_{l}v_{m,l}}{v_{m,m}}\right\rfloor, \nonumber \\
& ~~~~~~~~~~~m=1, 2, \ldots,n.
\end{eqnarray}
Since the fractional part in the above equation is of the form $s/q_m$, $s \in \{0,1,\ldots,q_m-1\}$, where $q_m$ is defined above, it follows that  $0 \leq s/q_m < 1$. Thus
\begin{eqnarray}
\lefteqn{{b}_m=} & \nonumber \\  & \left\{ \begin{array}{cc}
\tilde{b}_m-\left\lfloor \frac{\sum_{l=m+1}^{n}{b}_{l}v_{m,l}}{v_{m,m}}\right \rfloor, & s\leq s(m),  \\
\tilde{b}_m-\left\lfloor \frac{\sum_{l=m+1}^{n}b_{l}v_{m,l}}{v_{m,m}}\right \rfloor-1,  & s >  s(m).
\end{array}
\right.
\end{eqnarray}
can be computed in the fusion center $F$ in the order $m=n,n-1,\ldots,1$.
\end{proof}
\begin{corollary}\label{corocost}
The rate required to transmit $s(m)$, $m=1,2,\ldots,n-1$ is no larger than $\sum_{i=1}^{n-1}\log_2(q_i)$ bits.
\end{corollary}
Thus the total rate for computing the Babai point at the fusion center $F$ under the centralized model is no larger than $\sum_{i=1}^nh(p_i)-\log_2|\det V| -n \log_2(\alpha)+\sum_{i=1}^{n-1}\log_2(q_i)$ bits, where $h(p_i)$ is the differential entropy of random variable $X_i$, and scale factor $\alpha$ is small. Thus the incremental cost due to the $s(m)$'s does not scale with $\alpha$. However when $\alpha$ is small, this incremental cost can be considerable, if the lattice basis is not properly chosen as the following examples illustrate.

\begin{example} Consider the hexagonal $A_{2}$ lattice generated by 
$$V=\begin{pmatrix}
1 & \frac{1}{2} \\
0 & \frac{\sqrt{3}}{2}
\end{pmatrix}.$$ 
 The basis vectors forms an angle of $60^{\circ}$ and applying what we described above we have that the coefficients $b_{2}$ and $b_{1}$ are given respectively by
{\small \begin{equation}
b_{2}=\left[\frac{x_{2}}{v_{22}} \right]=\left[\frac{2}{\sqrt{3}}x_{2} \right]
\end{equation} }
and
{\small \begin{eqnarray}
b_{1}&=&\left[\frac{x_{1}}{v_{11}} - \left\{\frac{b_{2}v_{21}}{v_{11}} \right\} \right]-\left\lfloor\frac{b_{2}v_{21}}{v_{11}} \right\rfloor \\
&=& \left[ x_{1} -  \left\{\left[\frac{2}{\sqrt{3}}x_{2} \right] \frac{1}{2} \right\} \right] - \left\lfloor\left[ \frac{2}{\sqrt{3}}x_{2} \right] \frac{1}{2} \right\rfloor.
\end{eqnarray} }

	Hence, for any real vector $x=(x_1, x_2)$ we have $\left\{ \left[ \frac{2}{\sqrt{3}}x_{2} \right] \frac{1}{2} \right\}= \frac{s}{q}$, with $q=2$ and $s \in \{0,1\}$.  Node one must then send the largest integer $s(1)$ in the range $\{0,1\}$ for which $\left[x_{1}-\frac{s(1)}{q_{1}}\right]=[x_{1}]$ and $s(1)=0$ or $s(1)=1$ depending on the value that $x_{1}$ assumes.
	
	The cost of this procedure, according to Corollary \ref{corocost}, is no larger than $\log_{2}q_{1}=1$ bit. Thus the cost of constructing the nearest plane partition for the  hexagonal lattice is at most one bit.

\end{example}

\begin{example} Suppose a lattice generated by 
$$V=\begin{pmatrix}
1 & \frac{311}{1000} \\
0 & \frac{101}{100}
\end{pmatrix}.$$ 
One can notice that the basis vectors form an angle of approximately $72.89^{\circ}$ and they are already Minkowski-reduced. Using the theory developed above we have that
{\small \begin{equation}
b_{2}=\left[\frac{x_{2}}{v_{22}} \right]=\left[\frac{100}{101}x_{2} \right]
\end{equation} }
and
{\small \begin{eqnarray}
b_{1}&=&\left[\frac{x_{1}}{v_{11}} - \left\{\frac{b_{2}v_{21}}{v_{11}} \right\} \right]-\left\lfloor\frac{b_{2}v_{21}}{v_{11}} \right\rfloor \\
&=& \left[ x_{1} -  \left\{\left[\frac{100}{101}x_{2} \right] \frac{311}{1000} \right\} \right] - \left\lfloor\left[ \frac{100}{101}x_{2} \right] \frac{311}{1000} \right\rfloor.
\end{eqnarray} }

	Consider, for example, $x=(1, 1)$ then we have that $\left\{\left[\frac{100}{101}x_{2} \right] \frac{311}{1000} \right\}=\frac{311}{1000}=\frac{s}{q}.$ In this purpose, node one sends the largest integer $s(1)$ in the range $\{0,1, \dots, 999\}$ for which $\left[x_{1}-\frac{s(1)}{q_{1}}\right]=[x_{1}]$ and we get $s(1)=500.$ 
	
	This procedure will cost no larger than $\log_{2}q_{1}=\log_{2}1000 \approx 9.96$ and in the worst case, we need to send almost $10$ bits to achieve Babai partition in the centralized model. 

\end{example}

%
%

The analysis here points to the importance of the number-theoretic structure of the generator matrix  $V$ in determining the communication requirements for computing $x_{np}$. 

\subsection{Interactive Model}
For $i=n,n-1,\ldots,1$, node  $S_i$ sends $U_i=\left[ (X_i-\sum_{j=i+1}^n\alpha v_{ij}U_j)/\alpha v_{ii}\right]$ to all other nodes. The total number of bits communicated is given by $R=(n-1)\sum_{i=1}^n H(U_i|U_{i+1},U_{i+2},\ldots,U_n)$. For $\alpha$ suitably small, and under the assumption of independent $X_i$, this rate can be approximated by $R=(n-1)\sum_{i=1}^n h(p_i)-\log_2(\alpha v_{ii})$. Normalizing so that $V$ has unit determinant we get $R=(n-1)\sum_{i=1}^nh(p_i)-n(n-1) \log_2(\alpha)$.

\section{Conclusion and future work}	
\label{secCC}

	We have investigated the closest lattice point problem in a distributed network, under two communication models, centralized and interactive. By exploring the nearest plane (Babai) partition for a given Minkowski-reduced basis, we have determined the probability of error in analytic form in two dimensions. In two dimensions, the error depends on the ratio of the norms of these vectors and on the angle between them. We have also calculated the number of bits that nodes need to send in both models (centralized and interactive) to achieve the rectangular nearest plane partition. We have also demonstrated the importance of proper basis selection, for minimizing the probability of error. 	
	One question we are interested in is whether  similar results can be derived for greater dimensions. For example, it may be possible to generalize the results derived here to families $A_{n} $ and $D_{n}$ for which reduced form bases are already available.
	
\section{Acknowledgment}
CNPq (140797/2017-3, 312926/2013-8) and FAPESP (2013/25977-7) supported the work of MFB and SIRC. VV was supported by  CUNY-RF and CNPq (PVE 400441/2014-4).

%
%

\end{document}